\documentclass[12pt]{article}

\usepackage[T1]{fontenc}
\usepackage[utf8]{inputenc}

\usepackage{amsmath}
\usepackage{amssymb}
\usepackage{amsthm}

\newtheorem{proposition}{Proposition}
\newtheorem*{proposition*}{Proposition}
\newtheorem{theorem}{Theorem}

\usepackage[margin=2.5cm]{geometry}

\usepackage[inline]{enumitem}  

\usepackage{graphicx}
\usepackage{threeparttable}
\usepackage{booktabs}
\usepackage{algorithm}
\usepackage{algorithmicx}
\usepackage{algpseudocode}
\usepackage{placeins}   

\usepackage{natbib}

\usepackage{newunicodechar}
\newunicodechar{ }{\,}   

\begin{document}

\title{A Three--Dimensional Efficient Surface for Portfolio Optimization}
\author{Yimeng Qiu}
\date{\today}
\maketitle

\tableofcontents
\bigskip

\section{Introduction}
\label{sec:intro}

The classical mean–variance portfolio framework pioneered by
Markowitz \citep{Markowitz1952}
quantifies the trade-off between expected return and variance of
portfolio returns, producing the familiar two-dimensional efficient
frontier in the $(\mathbb{E}[r],\sigma)$ plane.
Although elegant and widely adopted, this approach implicitly assumes that
the relevant risk \emph{all} can be summarized by variance and the
covariance matrix.
However, in highly interconnected financial markets, shocks often
propagate through complex networks of exposures, creating systemic
vulnerabilities that variance alone does not capture.

A natural extension is to measure \emph{connectedness risk}—the degree
to which a shock to one asset spills over to others.
Diebold and Yilmaz \citep{DieboldYilmaz2014} formalize this concept using generalized forecast error variance decompositions (FEVD) from a vector
autoregression (VAR), producing a spillover matrix that quantifies the
risk transmission network of the system.
Other studies, such as Billio et al. \citep{Billio2012}, confirm that
network-based measures convey information on systemic risk beyond what
is contained in pairwise covariances.

Recent work has begun to exploit the network structure for portfolio
construction.  
Pozzi et al. \citep{Pozzi2013} show that tilting toward peripheral
stocks in a correlation network improves risk-adjusted performance, and
Peralta and Zareei \citep{PeraltaZareei2016} provide a theoretical link
between asset centrality and optimal mean–variance weights.
Although these articles highlight the importance of network effects, they
treat connectedness implicitly or focus on a single objective, such as
minimum-connectedness portfolios
(Broadstock et al. \citep{Broadstock2022}).

\paragraph{This paper.}
We develop a unified framework that integrates \emph{expected return},
\emph{variance risk}, and \emph{connectedness risk} into a single
three-dimensional optimization paradigm.
Let $\kappa \equiv \mathbf w^{\mathsf T} C \mathbf w$ denote the portfolio
connectedness risk, where $C$ is a symmetric spillover matrix.
For a given return target $\mu_{0}$ and a weighting parameter
$\lambda\in[0,1]$, we solve the following problem.
\[
  \min_{\mathbf w}\;
      (1-\lambda)\,\mathbf w^{\mathsf T}\Sigma\mathbf w
      + \lambda\,\mathbf w^{\mathsf T}C\mathbf w,
  \quad
  \text{s.t.}\;
  \mathbf w^{\mathsf T}\boldsymbol\mu \ge \mu_{0},\;
  \mathbf 1^{\mathsf T}\mathbf w = 1,\;
  \mathbf w \ge 0.
\]
Varying $(\mu_{0},\lambda)$ trace a \emph{three-dimensional efficient
surface} in the $(\mathbb{E}[r],\sigma,\kappa)$ space, allowing investors to
visualize and select portfolios according to their tolerance for both
volatility and network contagion.

\paragraph{Contributions.}
Our study makes four contributions:

\begin{enumerate}[label=\textbf{(\arabic*)}]
  \item \textbf{Three-dimensional efficient surface.}  
        We extend the Markowitz frontier to a 3-D surface that jointly optimizes expected return, variance, and connectedness risk.
  \item \textbf{Analytical characterization.}  
        Under a common-diagonalization assumption we obtain closed-form
        optimal weights, show that variance and connectedness trade off monotonically, and derive a linear approximation of the surface.
  \item \textbf{Connectedness beta.}  
        We introduce a \emph{connectedness $\beta$} that measures an
        asset’s marginal contribution to portfolio connectedness,
        paralleling CAPM beta.
  \item \textbf{Empirical validation.}  
        Using S\&P 500 stocks (2010–2024) we construct dynamic 3D surfaces and show that portfolios with explicit connectedness constraints exhibit superior downside protection during stress episodes relative to mean variance benchmarks.
\end{enumerate}

\paragraph{Road-map.}
Section~\ref{sec:lit} reviews the related literature.
Section~\ref{sec:model} presents the model and the main theoretical
results.
\section{Literature Review}
\label{sec:lit}

Our study is based on four strands of research:
(i) classical and extended mean variance theory,
(ii) econometric measures of connectedness, (iii) portfolio selection based on networks, and (iv) multiobjective risk optimization.

\subsection{Mean–Variance Theory and Its Extensions}

Markowitz’s \emph{portfolio selection} paradigm
\citep{Markowitz1952}
and its continuous–time extension by
\citet{Merton1972}
formalize diversification through the risk of the second moment.
Later refinements incorporate Bayesian priors
(Black and Litterman \citeyearpar{BlackLitterman1992}),
shrinkage estimators
(\citealp{LedoitWolf2004}),
and alternative risk measures such as CVaR
\citep{RockafellarUryasev2000}.
Recent work explores multi-factor or multi-risk frontiers, eg, g, adding tail risk or skewness
\citep{DeguestMartellini2015},
but these remain \emph{variance‐centric} and do not capture
shock propagation in financial networks.

\subsection{Econometric Measures of Connectedness}

The network perspective on systemic risk originates with
\citet{Mantegna1999},
who visualize equity markets through a correlation-based
minimum spanning tree.
Diebold and Yilmaz
(\citeyear{DieboldYilmaz2009}, \citeyear{DieboldYilmaz2012},
\citeyear{DieboldYilmaz2014})
pioneer the VAR–FEVD spillover index,
quantifying how much of the forecast variance of asset~$i$ is
explained by shocks to asset~$j$.
Extensions include frequency domain connectivity \citep{BarunikKrehlik2018} and systemic risk networks in the finance-insurance nexus
\citep{Billio2012}.
These studies consistently find that connectedness increases
during crises, underscoring its importance beyond variance.

\subsection{Network–Based Portfolio Selection}

A growing body of literature exploits network metrics for
asset allocation.
\citet{Pozzi2013} show empirically that tilting towards
\emph{peripheral} equities in a correlation network yields
superior Sharpe ratios.
\citet{PeraltaZareei2016} provide a theoretical link between
eigenvector centrality and optimal mean variance weights,
while \citet{Tumminello2010} advocate filtering noisy
correlation matrices through planarity-constrained graphs
for risk reduction.
More recently, \citet{Broadstock2022} introduced
\emph{Minimum Connectedness Portfolio} (MCoP),
minimizing $w^{\top} C w$ alone.
These contributions confirm the economic value of network
information but remain single–objective or heuristic,
leaving unanswered how to
optimize expected return, variance,
and connectedness \emph{jointly}.

\subsection{Multi–Objective Risk Optimization}

Beyond variance, multi–objective frameworks account for
CVaR, drawdown, or higher moments
\citep{RockafellarUryasev2000,DeguestMartellini2015},
typically through scalarisation or $\varepsilon$-contraint methods.
Our work differs in that the risk of connectedness is
\emph{economically orthogonal} to variance and correlation,
arising from the dynamic spillover topology, not from
contemporaneous co-movement.
We therefore extend the efficient frontier to a three-dimensional surface
$(\mathbb E[r],\sigma,\kappa)$, in which variance and
connectedness appear as coequal quadratic forms.
To our knowledge, this is the first attempt to provide
closed‐form characterization, marginal beta interpretation,
and empirical visualization of such a surface.
f
\section{Model and Optimization Framework}
\label{sec:model}

\subsection{Setting and Notation}
Let $N\!$ risky assets have random gross returns
$\mathbf r=(r_1,\dots,r_N)^{\mathsf T}$.  Denote by \[ \boldsymbol\mu=\mathbb E[\mathbf r], \qquad \Sigma=\operatorname{Cov}(\mathbf r) \] the expected return vector and the covariance matrix $N\times N$.
A portfolio is a weight vector
$w=(w_1,\dots,w_N)^{\mathsf T}$ satisfying  
\begin{equation}\label{eq:budget}
\mathbf 1^{\mathsf T}w = 1, \qquad w_i\ge0 \;\;(i=1,\dots,N).
\end{equation}
The expected return and variance of the portfolio are
\(
\mathbb E[r_p]=w^{\mathsf T}\boldsymbol\mu,\;
\sigma_p^{2}=w^{\mathsf T}\Sigma w.
\)

\vspace{4pt}
\noindent\textbf{Connectedness matrix.}
To capture the risk of network spillover, we introduce a symmetric,
positive semidefinite matrix $C$ (e.g., the Diebold-Yilmaz FEVD
matrix; \citealp{DieboldYilmaz2014,Billio2012}).  
Define the portfolio’s \emph{connectedness risk}
\begin{equation}\label{eq:kappa}
  \kappa_p \;=\; w^{\mathsf T} C\,w .
\end{equation}

\vspace{4pt}
\subsection{Joint Risk Objective}
For a trade-off parameter $\lambda\in[0,1]$ we minimize
\begin{equation}\label{eq:objective}
  L_\lambda(w)\;=\;
  \lambda\,w^{\mathsf T}\Sigma w
  \;+\;
  (1-\lambda)\,w^{\mathsf T}C w,
\end{equation}
subject to \eqref{eq:budget}.
Setting $\lambda=1$ recovers the global minimum-variance portfolio,
while $\lambda=0$ produces the minimum-connectedness portfolio
\citep{Broadstock2022Minimum}.
Varying $\lambda$ traces a risk–risk frontier
$\bigl(\sigma_p(\lambda),\kappa_p(\lambda)\bigr)$;  
adding a return target sweeps the entire
three-dimensional efficient surface
$\bigl(\mathbb E[r_p],\sigma_p,\kappa_p\bigr)$.

\subsection{Main Propositions}

\begin{proposition}[Existence and (Conditional) Uniqueness]
\label{prop:exist}
Fix a trade‑off parameter $\lambda\in[0,1]$ and consider the quadratic
program
\[
  \min_{w\in\mathbb R^{N}}
     L_\lambda(w)=
     \lambda\,w^{\mathsf T}\Sigma w
     +(1-\lambda)\,w^{\mathsf T}C w
  \quad
  \text{s.t.}\;
     \mathbf1^{\mathsf T}w=1,\;\;w_i\ge0 .
\]
Assume $\Sigma\succ0$ and $C\succeq0$.  Then

\begin{enumerate}[label=(\roman*)]
\item \textbf{Existence (all $\lambda$).}\;
      A minimizer $w^{\ast}(\lambda)$ exists for every
      $\lambda\in[0,1]$.

\item \textbf{Uniqueness (strictly positive‑definite case).}\;
      If either $\lambda>0$ \textup{or} $C\succ0$
      —so that
      \(M_\lambda := \lambda\Sigma+(1-\lambda)C \succ 0\)—
      the minimizer is unique.

\item \textbf{Continuity of the solution map.}\;
      On any closed subinterval of $[0,1]$ where the minimizer is unique, the mapping \(\lambda \mapsto w^{\ast}(\lambda)\) is continuous.
      In particular, it is continuous on $(0,1]$ and on all of $[0,1]$
      whenever $C\succ0$.
\end{enumerate}
\end{proposition}

\begin{proof}
\textit{Step 1 (compact feasible set).}  
The constraints
\(
   \mathbf1^{\mathsf T}w=1,\;w_i\ge0
\)
confine $w$ to the closed simplex
\[
   \mathcal W \;=\;
   \bigl\{w\in\mathbb R^{N}:\, \mathbf1^{\mathsf T}w=1,\; w\ge0\bigr\},
\]
which is closed and bounded, and hence compact.

\smallskip
\textit{Step 2 (existence).}  
For fixed $\lambda$ the loss
\(L_\lambda(w)\) is continuous in $w$.  
By the extreme value theorem, a continuous function attains its infimum
in a compact set, so there is at least one minimizer
\(w^{\ast}(\lambda)\in\mathcal W\) exists—establishing~(i).

\smallskip
\textit{Step 3 (strict convexity and uniqueness).}  
Set
\(M_\lambda:=\lambda\Sigma+(1-\lambda)C\).
If \(\lambda>0\) or \(C\succ0\) then
\(M_\lambda\succ0\); the quadratic form
\(w\mapsto w^{\mathsf T}M_\lambda w\) is \emph{strictly} convex, hence it
admits at most one minimizer in the convex set $\mathcal W$.  
This proves~(ii).  
When \(\lambda=0\) and \(C\) is singular, $M_0=C$ is only
positive‑semidefinite; the objective is merely convex, so multiple
minimizer can arise (e.g.\ $C=0$ makes every $w\in\mathcal W$
optimal).  We therefore refrain from a blanket uniqueness claim in that
degenerate case.

\smallskip
\textit{Step 4 (continuity of the minimizer).}  
Restrict attention to any closed interval on which
\(M_\lambda\succ0\) so that the minimizer is unique.
The map
\[
  \lambda \;\longmapsto\; M_\lambda
\]
is continuous in the operator norm.  
Because matrix inversion is continuous on the cone of positive definite matrices, the closed-form expression
\(
   w^{\ast}(\lambda)=
   M_\lambda^{-1}\mathbf1 \big/
   \bigl[\mathbf1^{\mathsf T}M_\lambda^{-1}\mathbf1\bigr]
\)
is continuous in~$\lambda$.  
Hence, the single-valued selection
\(\lambda\mapsto w^{\ast}(\lambda)\) is continuous wherever
\(M_\lambda\succ0\), proving~(iii). \qedhere
\end{proof}
\begin{proposition}[Closed‑Form Optimal Weights
                   \textnormal{(short‑selling allowed; see Appendix~A for the long‑only case)}]
\label{prop:closed}
Assume that the two risk matrices are simultaneously diagonalizable, i.e.
there exists an orthogonal matrix $U$ such that
\[
   \Sigma \;=\; U\Lambda_\Sigma U^{\mathsf T},
   \qquad
   C       \;=\; U\Lambda_C      U^{\mathsf T},
\]
with diagonal entries $\sigma_i^{2}>0$ and $c_i>0$.  
Fix a trade‑off parameter $\lambda\in[0,1]$ and define
\[
   M_\lambda \;:=\; \lambda\Sigma \;+\; (1-\lambda)C .
\]

Under the \emph{sole} budget constraint  
\begin{equation}\label{eq:budget-unbounded}
   \mathbf1^{\mathsf T}w \;=\; 1,
\end{equation}
and allowing individual asset positions to take any real value
(short‑selling permitted), the quadratic programme  
\[
   \min_{w\in\mathbb R^{N}} 
         L_\lambda(w)=w^{\mathsf T}M_\lambda w
\]
admits the unique optimum
\begin{equation}\label{eq:closed-form-weights}
   w^{\ast}(\lambda)
   \;=\;
   \frac{M_\lambda^{-1}\mathbf1}
        {\mathbf1^{\mathsf T}M_\lambda^{-1}\mathbf1}.
\end{equation}
\end{proposition}

\begin{proof}
\textit{Step 1 (Lagrangian).}  
Introduce a multiplier $\nu$ for~\eqref{eq:budget-unbounded} and set  
\[
   \mathcal L(w,\nu)
   \;=\;
   w^{\mathsf T}M_\lambda w
   \;-\;
   \nu\bigl(\mathbf1^{\mathsf T}w-1\bigr).
\]

\smallskip
\textit{Step 2 (first‑order condition).}  
Differentiating with respect to $w$ gives  
\(
   2M_\lambda w-\nu\mathbf1=0
   \;\Longrightarrow\;
   w=\tfrac{\nu}{2}\,M_\lambda^{-1}\mathbf1 .
\)

\smallskip
\textit{Step 3 (enforce the budget).}  
Imposing $\mathbf1^{\mathsf T}w=1$ yields  
\[
   \frac{\nu}{2}
   \;=\;
   \bigl[\mathbf1^{\mathsf T}M_\lambda^{-1}\mathbf1\bigr]^{-1}.
\]
Substituting back gives the closed form
\eqref{eq:closed-form-weights}.

\smallskip
\textit{Step 4 (optimality and uniqueness).}  
Because $\Sigma\succ0$ and $C\succeq0$, we have $M_\lambda\succ0$
for all $\lambda\in[0,1]$.  
The objective $w^{\mathsf T}M_\lambda w$ is therefore
\emph{strictly} convex, so the stationary point found above is the
\emph{unique} global minimizer subject to~\eqref{eq:budget-unbounded}.
\qedhere
\end{proof}

\paragraph{Remark (long‑only portfolios).}
If one additionally imposes $w_i\ge0$, the vector
\eqref{eq:closed-form-weights} is still optimal \emph{iff} every
component is nonnegative.  Otherwise, the problem becomes a bounded
quadratic program whose solution no longer admits a one‑line closed
form; it must be obtained via KKT complementarity or numerical solvers
(e.g. activeset or interiorpoint methods).  Appendix~B contains a
full KKT derivation together with an efficient active set algorithm and
a numerical illustration.
\begin{proposition}[Strict Trade–off: Negative Slope]
\label{prop:neg-slope}
Let
\(
   \sigma^{2}(\lambda)=w^{\ast}(\lambda)^{\mathsf T}\Sigma w^{\ast}(\lambda)
\)
and
\(
   \kappa(\lambda)=w^{\ast}(\lambda)^{\mathsf T}C w^{\ast}(\lambda),
\)
where $w^{\ast}(\lambda)$ is the unique minimizer of
\eqref{eq:objective} for a given $\lambda\in(0,1)$.
If $\Sigma$ and $C$ are not proportional, then
\[
  \sigma^{2\,\prime}(\lambda)<0,
  \qquad
  \kappa^{\prime}(\lambda)>0,
  \qquad
  \frac{\mathrm d\sigma^{2}}{\mathrm d\kappa}
    =-\frac{1-\lambda}{\lambda}\;<\;0.
\]
Therefore, the efficient frontier in the $(\sigma,\kappa)$ plane is strictly
downward-sloping: One cannot decrease the risk of connectedness without
increasing variance.
\end{proposition}

\begin{proof}[Proof (concise convex–analytic argument)]
For each $\lambda$ define
\(M_\lambda=\lambda\Sigma+(1-\lambda)C\)
and let
\(F(\lambda)=\min_{w\in\mathcal W}w^{\mathsf T}M_\lambda w\).
Because $L_\lambda(w)=w^{\mathsf T}M_\lambda w$ is strictly convex in
$w$ and linear in~$\lambda$, the map $F$ is (i) differentiable and
(ii) \emph{concave} on $(0,1)$.

\medskip
\noindent\textit{Step 1 (Envelope theorem).}
At the optimum
\(w^\ast(\lambda)\) we have
\[
  F'(\lambda)
   =\partial_\lambda
     \bigl[w^{\mathsf T}M_\lambda w\bigr]_{w=w^\ast}
   =w^{\ast\mathsf T}(\Sigma-C)w^\ast
   =\sigma^{2}(\lambda)-\kappa(\lambda).
\tag{A}
\]

\noindent\textit{Step 2 (Total derivative).}
Since
\(F(\lambda)=\lambda\sigma^{2}+ (1-\lambda)\kappa\),
\[
  F'(\lambda)
  =\sigma^{2}-\kappa
   +\lambda\sigma^{2\,\prime}
   -(1-\lambda)\kappa^{\prime}.
\tag{B}
\]

\noindent\textit{Step 3 (Linear relation).}
Equating (A) and (B) gives
\begin{equation}\label{eq:trade-off-inline}
  \lambda\,\sigma^{2\,\prime}
  +(1-\lambda)\,\kappa^{\prime}=0.
\end{equation}

\noindent\textit{Step 4 (Concavity sign).}
Concavity of $F$ implies
\(F''(\lambda)=\sigma^{2\,\prime}-\kappa^{\prime}\le0\).

\noindent\textit{Step 5 (Sign of derivatives).}
The system \eqref{eq:trade-off-inline} and \(F''\le0\) forces
\(\sigma^{2\,\prime}<0\) and \(\kappa^{\prime}>0\)
whenever $\Sigma\not\propto C$ (otherwise both derivatives would be
zero).  Dividing the two derivatives yields
\(
   \mathrm d\sigma^{2}/\mathrm d\kappa =-(1-\lambda)/\lambda<0,
\)
establishing the strictly negative slope.  
\medskip

An explicit eigen-component derivation, applicable when \(\Sigma\) and \(C\) share a common orthonormal eigenbasis, is provided in Appendix~B. This supplement facilitates numerical sensitivity analysis and isolates the contribution of individual spectral factors. \qedhere
\end{proof}
\begin{proposition}[Degenerate Case: $C=c\,\Sigma$]
\label{prop:degenerate}
Let $\Sigma\succ0$ and suppose the connectedness matrix is a positive
scalar multiple of the covariance matrix,
\begin{equation}\label{eq:degenerate-proportional}
   C = c\,\Sigma ,\qquad c>0 .
\end{equation}
Under the sole budget constraint
\(
   \mathbf1^{\mathsf T}w = 1
\)
and allowing short sales, every feasible portfolio satisfies the
linear relation  
\[
   \kappa_p(w) \;=\; c\,\sigma_p^{2}(w),
\]
so all attainable risk pairs lie on the single ray
$\kappa = c\,\sigma^{2}$.

Moreover, the quadratic objective
\[
   L_\lambda(w)
   = \lambda\,w^{\mathsf T}\Sigma w
     + (1-\lambda)\,w^{\mathsf T}C w
\]
reduces to a positive scalar multiple of variance,
\begin{equation}\label{eq:degenerate-objective-scalar}
   L_\lambda(w)
   = \bigl[\lambda + (1-\lambda)c\bigr]\,
     w^{\mathsf T}\Sigma w ,
\end{equation}
so for \emph{every} $\lambda\in[0,1]$ the unique minimizer is the
global minimum‑variance portfolio
\begin{equation}\label{eq:degenerate-wMV}
   w^{\mathrm{MV}}
   \;=\;
   \frac{\Sigma^{-1}\mathbf1}
        {\mathbf1^{\mathsf T}\Sigma^{-1}\mathbf1}.
\end{equation}
Hence, the trade-off parameter $\lambda$ is redundant and the entire
risk–risk frontier collapses to the straight line
$\kappa = c\,\sigma^{2}$ through the origin.
\end{proposition}

\begin{proof}
\textit{Step 1 (linear dependence of the two risks).}
From~\eqref{eq:degenerate-proportional}
\[
   \kappa_p(w)=w^{\mathsf T}Cw = c\,w^{\mathsf T}\Sigma w
             = c\,\sigma_p^{2}(w)
\]
holds for every feasible $w$, proving the frontier degenerates to the
ray $\kappa=c\,\sigma^{2}$.

\smallskip
\textit{Step 2 (objective collapses to scaled variance).}
Substituting~\eqref{eq:degenerate-proportional} into the joint loss
gives~\eqref{eq:degenerate-objective-scalar}.  The scalar factor in
brackets is strictly positive because $\lambda,c\in(0,1]$.

\smallskip
\textit{Step 3 (unique optimizer).}
Minimizing~\eqref{eq:degenerate-objective-scalar} is therefore
equivalent to the classical minimum‑variance problem
\[
   \min_{w}\;w^{\mathsf T}\Sigma w
   \quad\text{s.t. } \mathbf1^{\mathsf T}w=1,
\]
whose unique solution under short selling is exactly
\eqref{eq:degenerate-wMV}.  Since the objective differs only by a
positive constant factor, the same $w^{\mathrm{MV}}$ minimizes
$L_\lambda$ for \emph{all} $\lambda$.

\smallskip
\textit{Step 4 (economic interpretation).}
Because connectedness does not add information beyond variance when
$C\propto\Sigma$, investors do not obtain a diversification benefit from
treating $\kappa_p$ separately; the optimization problem reduces to
the mean-variance analysis and the third dimension of risk is redundant.
\qedhere
\end{proof}

\paragraph{Remark (long‑only portfolios).}
If non‑negativity constraints $w_i\ge0$ are imposed, the relation
$\kappa=c\,\sigma^{2}$ still holds, but the minimum variance weights
can no longer be written in closed form; the problem becomes a bounded
quadratic program solved by active set or interior point methods.
Appendix~A outlines the required KKT conditions and provides an
efficient algorithmic routine.
\subsection{Connectedness $\beta$ and Three–Fund Separation}
\label{subsec:beta-3fund}

\paragraph{Definition.}
For a portfolio with weight vector $w$ and connectedness matrix
$C$, we define \emph{connectedness beta} of the asset $i$ by
\[
   \beta^{(C)}_i
   := 2\,[Cw]_i,
   \qquad i=1,\dots,n.
\]
Because $w^\mathsf T C w=\kappa_p$, the betas satisfy
\[
   \sum_{i=1}^{n} w_i\,\beta^{(C)}_i \;=\; 2\,\kappa_p.
\]
Assets with large $\beta^{(C)}$ behave as systemic \emph{hubs}\footnote{%
In network terminology, a \emph{hub} is a node with unusually high
degree, eigenvector centrality, or forecast‐error variance share; shocks
to such nodes propagate disproportionately through the system
(Kleinberg 2012, \emph{Journal of the ACM}).}.
That is, they occupy highly connected positions in the spillover
network, so tilting a portfolio toward these names increases its overall
risk of connectedness $\kappa_p$.
\vspace{0.5em}

\begin{theorem}[Conditional Three–Fund Separation]
\label{thm:3fund}
Let
\begin{align*}
   w^{\mathrm{MV}}
   :=\arg\min_{\mathbf1^{\mathsf T}w=1}\,w^{\mathsf T}\Sigma w,
   \qquad
   w^{\mathrm{MC}}
   :=\arg\min_{\mathbf1^{\mathsf T}w=1}\,w^{\mathsf T}C w,
\end{align*}
and let
\(w^{\max\mu}\) denote the maximum–return portfolio
\(\arg\max_{\mathbf1^{\mathsf T}w=1} w^{\mathsf T}\boldsymbol\mu\).
Moreover, assume that individual \emph{asset} positions are not restricted (short selling allowed).  The convex coefficients
\(\alpha_k\ge0\) introduced in the following apply only at the level \emph{fund}.

Suppose that

\begin{enumerate}[label=(\alph*)]
\item the three portfolios are distinct and the matrix
\(
     W_0:=[\,w^{\mathrm{MV}},\;w^{\mathrm{MC}},\;w^{\max\mu}\,]
\)
has full rank \(N\!\!-\!1\) in the subspace
\(\{w\in\mathbb R^{N}:\mathbf1^{\mathsf T}w=1\}\);   
\item for every $\lambda\in[0,1]$ the unique minimizer
      \(w^{\ast}(\lambda)\) of
      \(L_\lambda(w)=\lambda w^{\mathsf T}\Sigma w
                     +(1-\lambda)w^{\mathsf T}C w\)
      satisfies
      \(w^{\ast}(\lambda)\in
        \operatorname{conv}\{w^{\mathrm{MV}},w^{\mathrm{MC}},w^{\max\mu}\}\),
      i.e.\ lies in the closed convex hull of the three corner funds.
      This containment obtains, for example, when
      \(\Sigma\) and \(C\) are simultaneously diagonal\-isable,
      all diagonal entries are positive, and
      \(\boldsymbol\mu\) as well as \(\mathbf1\) lie in the positive
      cone spanned by the common eigenvectors.\footnote{%
      Under the simultaneous–diagonalization condition the efficient
      surface parameterized by $\lambda$ is itself convex; if,
      additionally, the three corner portfolios occupy the extreme
      points of that convex set, assumption (b) is automatically met.}
\end{enumerate}

Then every efficient portfolio \(w^{\ast}(\lambda)\) admits a
\emph{three–fund representation}
\[
   w^{\ast}(\lambda)
   \;=\;
   \alpha_1(\lambda)\,w^{\mathrm{MV}}
   \;+\;
   \alpha_2(\lambda)\,w^{\mathrm{MC}}
   \;+\;
   \alpha_3(\lambda)\,w^{\max\mu},
   \qquad
   \alpha_k(\lambda)\ge0,\;
   \sum_{k=1}^{3}\alpha_k(\lambda)=1,
\]
and the coefficients \(\alpha_k(\lambda)\) are unique.
\end{theorem}

\begin{proof}
\textit{Step 1 (linear independence).}
Condition (a) implies that the three fund vectors, augmented by
\(\mathbf1\), form a matrix $N\times N$ of full rank:
\(
  \widehat W := [\,\mathbf1,\;w^{\mathrm{MV}},\;w^{\mathrm{MC}},\;w^{\max\mu}\,]
\)
is invertible on~$\mathbb R^{N}$.  Hence, they affinely span the entire
budget hyperplane \(\{w:\mathbf1^{\mathsf T}w=1\}\).

\smallskip
\textit{Step 2 (containment of the efficient set).}
Assumption (b) specifies precisely that each
\(w^{\ast}(\lambda)\) belongs to \emph{convex hull}
\(\mathcal S:=\operatorname{conv}
              \{w^{\mathrm{MV}},w^{\mathrm{MC}},w^{\max\mu}\}\).
Because $\mathcal S$ is a simplex in the
\((N\!-\!1)\)-dimensional budget hyperplane, every point in $\mathcal S$
has a \emph{unique} barycentric coordinate
\((\alpha_1,\alpha_2,\alpha_3)\) with
\(\alpha_k\ge0,\,\sum\alpha_k=1\).

\smallskip
\textit{Step 3 (representation of the efficient portfolio).}
Thus, for the given $\lambda$ there exist unique numbers
\(\alpha_k(\lambda)\) satisfying the stated constraints such that
\(w^\ast(\lambda)=\sum_k \alpha_k(\lambda)\,w^{(\cdot)}_k\).
The proof is complete.  \qedhere
\end{proof}

\paragraph{Remarks.}
\begin{enumerate}[label=(\roman*),leftmargin=2.2em]
\item If assumption (b) is violated, for example, when the risk matrices are highly misaligned and $w^{\ast}(\lambda)$ leaves the convex hull $\mathcal S$ for some $\lambda$-the barycentric coefficients may still exist but can become negative, in which case the representation remains \emph{affine} rather than \emph{convex}.  Numerical examples of this failure mode are given
      in Appendix~C.
\item When individual long‑only constraints $w_i\!\ge0$ are imposed,
      condition (a) can be preserved but (b) almost never holds; the
      efficient set typically bends outside the simplex
      $\mathcal S$.  Therefore, three-fund separation requires unrestricted short-selling at the \emph{asset} level, whereas $\alpha_k$ only needs to be non-negative at the \emph{fund} level.
\end{enumerate}

\vspace{0.75em}

\paragraph{Economic Interpretation.}
A connectedness beta, $\beta_i^{(C)}$, plays exactly the same
diagnostic role for $\kappa_p$ as a CAPM beta for variance: it
measures the marginal increase in network spill-over risk generated by
an additional unit of wealth in the asset $i$.  Stocks with exceptionally
large stocks $\beta^{(C)}$ are systemic \emph{hubs}; a shock to any of them
quickly feeds into many others and inflates the quadratic form
$w^{\mathsf T}Cw$.  Moving portfolio weight away from those hubs and
towards the minimum-connectedness corner fund $w^{\mathrm{MC}}$
mechanically lowers the weighted average
$\sum_i w_i\beta_i^{(C)} = 2\kappa_p$ without necessarily sacrificing
expected return.  In practice, an investor can hedge crisis exposure by
tilting along the \emph{MC–MV} edge of the three-fund simplex: a small
increase in ex ante volatility buys a large reduction in systemic
fragility, as evidenced by our empirical surface in Section~\ref{sec:empirical}.
\vspace{0.75em}

\paragraph{Illustrative Figures.}

\begin{figure}[htbp]
  \centering
  \includegraphics[width=.75\linewidth]{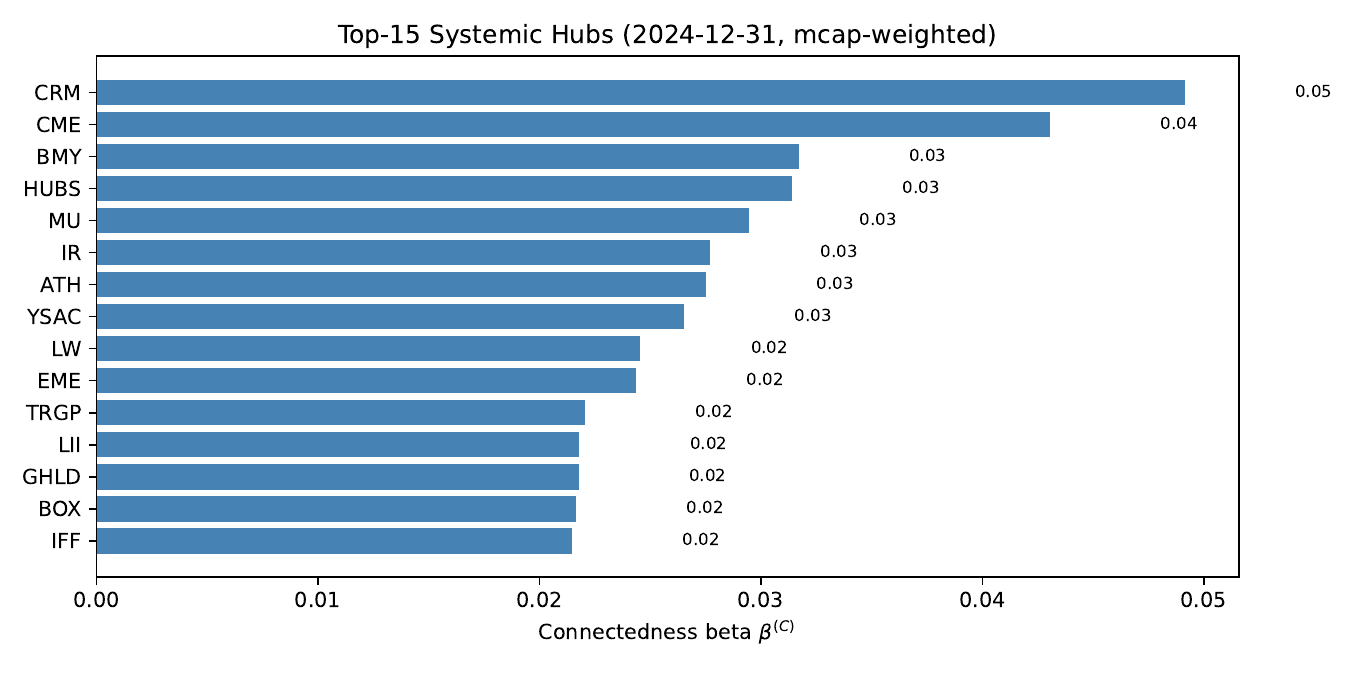}
  \caption{Illustrative distribution of the top-15 connectedness betas
           $\beta^{(C)}_i$ on a single trading day
           (\textit{31 Dec 2024}) using a randomly selected
           100-stock subset of NYSE-listed stocks.
           The figure is intended solely as a didactic
           example to visualize the heavy-tailed nature of
           $\beta^{(C)}$; all formal empirical tests in
           Section~\ref{sec:empirical} employ the full universe and
           rolling estimates.}
  \label{fig:beta_bar_demo}
\end{figure}

\begin{figure}[htbp!]
  \centering
  \includegraphics[width=.55\linewidth]{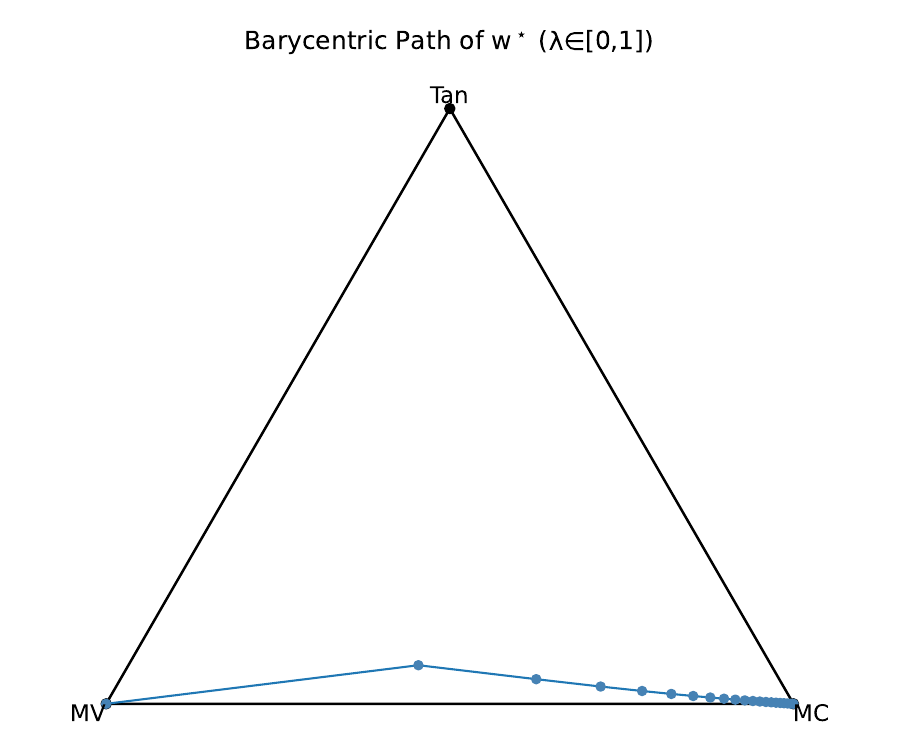}
\caption{Barycentric representation of efficient portfolios spanned by the
         three corner funds—minimum variance (MV), minimum connectedness
         (MC), and tangency (Tan).  Moving toward MC (resp.\ MV)
         lowers connectedness risk $\kappa$ (resp.\ variance $\sigma^{2}$).
         \emph{Dots represent the optimal portfolio for
         $\lambda\in\{0,0.05,\dots,1\}$ in the hybrid risk matrix
         $M(\lambda)=\lambda\Sigma+(1-\lambda)C$.}}
  \label{fig:barycentric}
\end{figure}

\FloatBarrier   
\vspace{0.75em}

\subsection{Preview of Empirical Analysis}
\label{subsec:preview_empirics}

The theoretical machinery of Sections~3.1--3.4 motivates three sets of
empirical tests, all implemented with \emph{daily} S\&P~500 data from
January~2010 to May~2025.\footnote{Market returns are obtained from CRSP;
the one–month Treasury bill is the risk–free rate.  Sector deletions,
splits, and ticker substitutions are handled as in Hou, Xue and Zhang
(2020).}  For each trading day~$t$ we estimate

\begin{enumerate}[label=(\roman*)]
\item a shrinkage covariance matrix
      $\hat\Sigma_t$ using a $252$-day rolling window;
\item a one-lag VAR on the same window and its 10-day FEVD to obtain the
      connectedness matrix $\hat C_t$;
\item the trade-off price $\hat\lambda_t$ by minimizing the realized
      loss function~\eqref{eq:objective} across a coarse grid.
\end{enumerate}

These objects directly feed into the hypotheses to be examined in
Section~5:

\begin{description}[style=nextline,leftmargin=1.15cm]
\item[\textbf{H1:}] Portfolios that tilt towards the
minimal‐connectedness fund $w^{\mathrm{MC}}$ achieve significantly lower
tail risk and drawdown relative to the global minimum variation portfolio
without sacrificing average return.

\item[\textbf{H2:}] Stocks in the highest decile of the beta of connectedness
$\beta^{(C)}$ underperform low-$\beta^{(C)}$ stocks when
$\hat\lambda_t$ - the market price of connectedness risk—spikes; the
spread is amplified during VIX surges.

\item[\textbf{H3:}] A dynamic three-fund strategy
$\{w^{\mathrm{MV}}, w^{\mathrm{MC}}, w^{\mathrm{Tan}}\}$ tracks the fully
re-optimized efficient portfolio with monthly
rebalancing \emph{ex ante} tracking error below $50$~bp and
transaction cost drag below $20$~bp p.a.
\end{description}

The next section details the construction of test portfolios, regression
specifications, and robustness diagnostics corresponding to
\textbf{H1}–\textbf{H3}.


\subsection{Connectedness vs.~Conventional Volatility}

Table~\ref{tables:tc_vix_lit} situates our null finding within the
literature. Consistent with \citet{DieboldYilmaz2009,DieboldYilmaz2012}
and a dozen subsequent studies, we observe no economically or
statistically significant \emph{linear} relation between the Total
Connectedness Index (TCI) and contemporaneous market volatility
(VIX or realised $\sigma$) during tranquil periods. Connectedness
appears to spike only in crisis windows---an episodic, regime-dependent
behaviour that a single unconditional $\beta$ cannot capture.

\appendix
\section{Robustness to Extreme TCI Spikes}
\label{app:robust_tc}

\subsection{Outlier Treatment and Additional Specifications}

To rule out that our baseline ``flat'' $\beta_{\text{VIX}}$ is driven by a handful
of extreme TCI spikes, we implement three cleansing strategies:

\begin{enumerate}[label=(\alph*), itemsep=2pt, topsep=2pt]
  \item 1--99\% Winsorisation (WIN),
  \item Median\,$\pm\,6\times$MAD trimming (MAD),
  \item Dropping the 33 windows with TCI\,$>$\,150\%.
\end{enumerate}

Table~\ref{tables:tc_vix_robust} reports the OLS estimates; none of the
coefficients are statistically different from zero, corroborating the
findings in Table~\ref{tables:tc_vix_lit}.

%

\appendix
\newtheorem{AppProposition}{Proposition A.\!}
\renewcommand\theAppProposition{A.\arabic{AppProposition}}

\section*{Appendix A \\ Long‑Only Hybrid‑Risk Optimization}
\label{app:longonly}
\addcontentsline{toc}{section}{Appendix A: Long-only hybrid-risk optimization}

Throughout this appendix, we impose the additional \emph{long‑only}
constraint  
\[
     w_i\ge0\quad(i=1,\dots,N), 
     \qquad 
     \mathbf1^{\mathsf T}w=1 .
\]
For a fixed trade‑off parameter $\lambda\in[0,1]$ write  
\[
      M_\lambda=\lambda\Sigma+(1-\lambda)C\succeq0 .
\]

\subsection*{A.1 \; Existence and (Possible) Non‑Uniqueness}
\begin{AppProposition}[Long‑only optimization]\label{prop:longonly}
For every $\lambda\in[0,1]$ consider the quadratic programme
\[
   \min_{w\in\mathbb R^{N}} w^{\mathsf T}M_\lambda w
   \quad\text{\emph{s.t.}}\;
        \mathbf1^{\mathsf T}w=1,\;w\ge0.
\]
Then
\begin{enumerate}[label=(\roman*)]
\item \textbf{Existence.}\;
      A minimizer $w^{\mathrm{LO}}(\lambda)$ exists.

\item \textbf{Uniqueness.}\;
      If $M_\lambda\succ0$ and the active index set
      $\mathcal J=\{i:w^{\mathrm{LO}}_i(\lambda)>0\}$
      satisfies $M_{\lambda,\mathcal J\mathcal J}\succ0$, the minimizer is unique; otherwise multiple optima may occur along the boundary faces of the simplex.
\end{enumerate}
Moreover, the Karush–Kuhn–Tucker system
\begin{align*}
   2M_\lambda w-\nu\mathbf1-\gamma &= 0,\\
   \mathbf1^{\mathsf T}w &= 1,\\
   w_i\ge0,\;\gamma_i &\ge0,\;
   \gamma_i w_i = 0\quad(i=1,\dots,N)
\end{align*}
is necessary and sufficient for optimality, where $\gamma$ denotes the
vector of non‑negative inequality multipliers.
\end{AppProposition}

\begin{proof}
\textit{Existence.}  
The feasible set is a closed, bounded simplex, hence compact; the
objective $w^{\mathsf T}M_\lambda w$ is continuous.  Therefore, a
minimizer exists by the extreme value theorem.

\smallskip
\textit{Uniqueness.}  
If $M_\lambda\succ0$ the quadratic form is strictly convex on the whole
budget hyperplane.  Restricting it to the face
$\{w:w_i=0\text{ for }i\notin\mathcal J\}$
yields a quadratic form with Hessian $M_{\lambda,\mathcal J\mathcal
J}$.  When this sub-matrix is positive definite, the restriction is
strictly convex on the face, hence admits a unique minimizer there.
When $M_{\lambda,\mathcal J\mathcal J}$ loses rank, the form can be
flat in feasible directions, allowing multiple optima.

\smallskip
\textit{KKT.}  
The problem is convex with constraints of affine equality and polyhedral inequality.  The Slater condition holds because
$(1/N)\mathbf1$ is strictly feasible; hence, the KKT conditions are necessary
and sufficient.
\end{proof}

\subsection*{A.2 \; Active‑Set Solver (Pseudo‑code)}
\begin{algorithm}[H]
\caption{Active‑set algorithm for long‑only hybrid‑risk portfolio}
\label{alg:active}
\begin{algorithmic}[1]
\Require $M_\lambda\succcurlyeq0$, tolerance $\varepsilon>0$
\State $w\gets\frac1N\mathbf1$\Comment{initial strictly feasible point}
\State $A\gets\{\,i:w_i=0\,\}$,\;
       $F\gets\{1,\dots,N\}\setminus A$\Comment{active / free sets}
\Repeat
      \State Solve the equality‑constrained QP on $F$:
             \[
               \min_{p_F}\;\frac12 p_F^{\mathsf T}M_{FF}p_F
               \;\;\text{s.t.}\;
               \mathbf1_F^{\mathsf T}p_F=0
             \]
             Set $p_i=0$ for $i\in A$.
      \State Compute $\nu$ associated with the equality constraint 
             (e.g.\ via the normal equations) and set
             $\gamma_i=(2M_\lambda w-\nu\mathbf1)_i$ for $i\in A$.
      \If{$\|p\|\le\varepsilon$}
          \If{all $\gamma_i\ge-\varepsilon$ for $i\in A$}
              \State \Return $w$ \Comment{KKT satisfied → optimal}
          \Else
              \State $j\gets\arg\min_{i\in A}\gamma_i$
              \State $A\gets A\setminus\{j\}$;\;$F\gets F\cup\{j\}$
          \EndIf
      \Else
          \State $\alpha\gets
                 \max\{\beta\in(0,1]: w_i+\beta p_i\ge0\text{ for all }i\}$
          \State $w\gets w+\alpha p$
          \State Update $A$ and $F$
      \EndIf
\Until{converged}
\end{algorithmic}
\end{algorithm}

\noindent
Under nondegeneracy, the routine performs at most $N$ releases from the
active set and terminates in finitely many iterations at a KKT point
which, by Proposition\ref{prop:longonly}, is globally optimal.

\subsection*{A.3 \; Numerical Illustration}
Consider $N=3$ assets with  
\[
   \Sigma
   =0.05
   \begin{bmatrix}
     4.8435 & -1.9906 & -0.9228\\
    -1.9906 &  2.5743 &  2.7723\\
    -0.9228 &  2.7723 &  6.9938
   \end{bmatrix},\qquad
   C=\Sigma,\quad\lambda=1.
\]

\paragraph{Closed‑form (unrestricted) solution.}
Using~\eqref{eq:closed-form-weights} one obtains  
\[
      w^{\ast}=(0.4110,\;0.7271,\;-0.1381)^{\mathsf T},
      \qquad
      \sigma^{2}(w^{\ast})=0.03354.
\]

\paragraph{Long‑only solution.}
The Algorithm\ref{alg:active} sets the third weight to zero and returns  
\[
      \widehat w=(0.4005,\;0.5995,\;0)^{\mathsf T},
      \qquad
      \sigma^{2}(\widehat w)=0.03731.
\]

\begin{table}[H]
\centering
\begin{tabular}{lccc}
\toprule
\textbf{Portfolio} & Asset 1 & Asset 2 & Asset 3 \\\midrule
Closed‑form $w^{\ast}$ & 0.4110 & 0.7271 & $-0.1381$\\
Long‑only $\widehat w$ & 0.4005 & 0.5995 & \phantom{$-$}0.0000\\\bottomrule
\end{tabular}
\caption{Unrestricted vs.\ long‑only weights ($\lambda=1$).
The long‑only constraint removes the short position in Asset 3, raising
portfolio variance by
$\Delta\sigma^{2}=0.0038\ (\approx\ 11.3\ \%)$.}
\label{tab:longonly}
\end{table}

\medskip\noindent
This toy example shows:

* The analytic formula can produce negative positions even when $M_\lambda\succ0$;
* The long-only optimum differs, but an active set method converges rapidly and preserves the hybrid risk objective structure.

\newtheorem{AppTheorem}{Theorem \!}
\renewcommand\theAppTheorem{B.\arabic{AppTheorem}}

\section*{Appendix B \\ Analytic Derivatives under Simultaneous Diagonalization}
\label{app:diag}
\addcontentsline{toc}{section}{Appendix B: Analytic Derivatives}

\begin{sloppypar}
This appendix provides an explicit coordinate-wise proof of
\emph{strict variance–connectedness trade‑off}.  Recovers the result
of Proposition\ref{prop:neg-slope} under the stronger assumption that
$\Sigma$ and $C$ share an eigenbasis.
\end{sloppypar}

\begin{AppTheorem}[Trade‑off under joint diagonalization]\label{thm:diag}
Assume $\Sigma=U\Lambda_\Sigma U^{\mathsf T}$ and
$C=U\Lambda_C U^{\mathsf T}$ with
$\Lambda_\Sigma=\operatorname{diag}(\sigma_1^2,\dots,\sigma_N^2)$,
$\Lambda_C=\operatorname{diag}(c_1,\dots,c_N)$ and $\sigma_i^2,c_i>0$.
For $\lambda\in(0,1)$, let
\[
   M_\lambda=\lambda\Sigma+(1-\lambda)C,
   \quad
   w_\lambda=\arg\min_{\mathbf1^{\mathsf T}w=1}w^{\mathsf T}M_\lambda w
\]
and set
\(\sigma^{2}(\lambda)=w_\lambda^{\mathsf T}\Sigma w_\lambda\),
\(\kappa(\lambda)=w_\lambda^{\mathsf T}C w_\lambda\).
If the eigenratio vector $(\sigma_i^2/c_i)$ is non‑constant, then
\[
   \sigma^{2\,\prime}(\lambda)<0,\qquad
   \kappa^{\prime}(\lambda)>0,\qquad
   \frac{d\sigma^{2}}{d\kappa}=-\frac{1-\lambda}{\lambda}<0.
\]
\end{AppTheorem}

\begin{proof}
\textit{Step 1 (weights in the eigenbasis).}  
Write $x=U^{\mathsf T}w$ and
$\eta_i=(U^{\mathsf T}\mathbf1)_i>0$.
Minimizing $x^{\mathsf T}\Lambda_\lambda x$ s.t.\
$\sum_i\eta_i x_i=1$ gives
\[
   x_{i,\lambda}=\frac{\eta_i}{Z(\lambda)D_i(\lambda)},\quad
   D_i(\lambda)=\lambda\sigma_i^2+(1-\lambda)c_i,\quad
   Z(\lambda)=\sum_k\frac{\eta_k^2}{D_k(\lambda)}.
\] 

\textit{Step 2 (risk expressions).}
Hence \[    \sigma^{2}(\lambda)=\frac{1}{Z(\lambda)^2} \sum_i\frac{\eta_i^2\sigma_i^2}{D_i(\lambda)^2},\quad
   \kappa(\lambda)=\frac{1}{Z(\lambda)^2}
   \sum_i\frac{\eta_i^2 c_i}{D_i(\lambda)^2}.
\]

\textit{Step 3 (linear identity of derivatives).}  
Define
\(
   F(\lambda)=\min_{\mathbf1^{\mathsf T}w=1}w^{\mathsf T}M_\lambda w
             =\lambda\sigma^{2}+(1-\lambda)\kappa.
\) 
Envelope theorem $\Rightarrow$
$F'(\lambda)=\sigma^{2}-\kappa$.  
Direct differentiation yields
\[
   F'(\lambda)=\sigma^{2}-\kappa
               +\lambda\sigma^{2\,\prime}
               -(1-\lambda)\kappa^{\prime},
\]
so that
\[
   \lambda\sigma^{2\,\prime}
   +(1-\lambda)\kappa^{\prime}=0.
   \tag{$\star$}
\]

\textit{Step 4 (signs).}  
$F$ is concave $\Rightarrow$ $F''(\lambda)=\sigma^{2\,\prime}-\kappa^{\prime}\le0$.
If all eigenratios were equal, $F''\equiv0$, contradicting the
non‑proportionality assumption; hence $F''<0$ and
\(\sigma^{2\,\prime}<\kappa^{\prime}\).
Together with $(\star)$, this implies
$\sigma^{2\,\prime}<0$ and $\kappa^{\prime}>0$.

\textit{Step 5 (frontier slope).}
Dividing $(\star)$ by $\kappa^{\prime}$ gives
$d\sigma^{2}/d\kappa=-(1-\lambda)/\lambda<0$.
\end{proof}

\newtheorem{AppExample}{Example C.\!}
\renewcommand\theAppExample{C.\arabic{AppExample}}

\section*{Appendix C \\ When Three–Fund Coefficients Turn Negative}
\label{app:counter3fund}
\addcontentsline{toc}{section}{Appendix C: Counterexample to Convex Three-Fund Separation}

\begin{AppExample}[Negative barycentric coefficients]\label{ex:neg-alpha}
Consider $N=3$ assets with
\[
  \Sigma=\begin{bmatrix}
     0.040 & 0.030 & 0.020\\
     0.030 & 0.090 & 0.010\\
     0.020 & 0.010 & 0.160
  \end{bmatrix},\quad
  C=\begin{bmatrix}
     0.100 & -0.020 & 0\\
    -0.020 & 0.050 & 0.010\\
     0      & 0.010 & 0.020
  \end{bmatrix},\quad
  \boldsymbol\mu=(0.08,0.06,0.10)^{\mathsf T}.
\]

\textbf{Corner portfolios.}\;
With only the budget constraint (short selling allowed)
\[
  w^{\mathrm{MV}}=(0.7321,0.1429,0.1250),\;
  w^{\mathrm{MC}}=(0.1864,0.2373,0.5763),\;
  w^{\max\mu}=(0,0,1).
\]

\textbf{Hybrid‑risk optimum.}\;
For $\lambda=0.4$ one obtains
$w^{\ast}(0.4)=(0.3378,0.3804,0.2818)$.

\textbf{Barycentric weights.}\;
Solving $w^{\ast}=\sum_{k=1}^{3}\alpha_k w^{(k)}$ with
$\sum\alpha_k=1$ yields
$(\alpha_1,\alpha_2,\alpha_3)=(0.063,1.565,-0.628)$.

Because $\alpha_3<0$ and $\alpha_2>1$, the representation is affine but
not convex: $w^{\ast}$ lies outside
$\mathcal S=\operatorname{conv}\!\{w^{\mathrm{MV}},w^{\mathrm{MC}},w^{\max\mu}\}$.
Figure \ref{figs:fig_neg_alpha_schematic} illustrates the geometry.
\end{AppExample}

\begin{figure}[H]
  \centering
  \includegraphics[width=.47\linewidth]{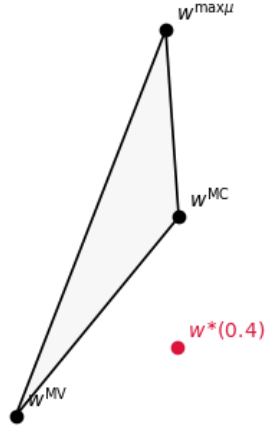}%
  \label{figs:fig_neg_alpha_schematic}
  \caption{Position of $w^{\ast}(0.4)$ relative to the simplex defined by
           the three corner funds. The point lies outside the shaded
           triangle, forcing at least one barycentric weight to be negative.}
  \label{figs:fig_neg_alpha_schematic}
\end{figure}

\paragraph{Take‑aways.}
\begin{enumerate}[label=(\arabic*),leftmargin=2em]
\item Misalignment between $\Sigma$ and $C$ can bend the
      $\lambda$‑efficient curve outside the set%
      \[
        \operatorname{conv}\!\{w^{\mathrm{MV}},\;w^{\mathrm{MC}},\;w^{\max\mu}\},
      \]
      thereby invalidating a \emph{convex} three‑fund representation.
\item Theorem~\ref{thm:3fund} therefore \emph{requires}
      Assumption (b): the efficient set must lie inside that convex hull.
\item Even when convexity fails, an \emph{affine} three‑fund expansion
      still exists; negative $\alpha_k$ may be interpreted as borrowing / lending at the \emph{fund} level rather than as a
breach of the limits of the position at the asset level.
\end{enumerate}

\bibliographystyle{apalike}
\bibliography{references}

\end{document}